\documentclass[runningheads]{llncs}

\usepackage{graphicx}
\usepackage{amsfonts}
\usepackage{amssymb}
\usepackage{algorithm}
\usepackage{algorithmic}
\usepackage{float}
\usepackage{enumerate}
\usepackage{pgfplots, tikz}
\usepackage{amsmath}
\usepackage{mathtools}
\DeclarePairedDelimiter{\floor}{\lfloor}{\rfloor}

\newlength\myindent
\begin{document}

\title{An Improved Exact Algorithm for the Exact Satisfiability Problem}

\author{Gordon Hoi}

\titlerunning{An Improved Exact Algorithm for the Exact Satisfiability Problem}
\authorrunning{G.~Hoi}

\institute{School of Computing, National University of Singapore,
13 Computing Drive, Block COM1, Singapore 117417, Republic of Singapore \\
\email{e0013185@u.nus.edu}}

\maketitle        

\begin{abstract}
\noindent
The Exact Satisfiability problem, XSAT, is defined as the problem of finding a satisfying assignment to
a formula $\varphi$ in CNF such that exactly one literal in each clause is assigned to be  ``1" and 
the other literals in the same clause are set to ``0". Since it is an important variant of the satisfiability problem,
XSAT has also been studied heavily and has seen numerous improvements 
to the development of its exact algorithms over the years.

The fastest known exact algorithm to solve XSAT runs in $O(1.1730^n)$ time, where $n$ is
the number of variables in the formula. In this paper, 
we propose a faster exact algorithm that solves the problem in $O(1.1674^n)$ time.
Like many of the authors working on this problem, we give a DPLL algorithm to solve it. The novelty
of this paper lies on the design of the nonstandard measure, to help us to tighten the analysis
of the algorithm further.

\end{abstract}

\section{Introduction}
\noindent 
Given a propositional formula $\varphi$ in conjunctive normal form (CNF), a common question to ask would be if
there is a satisfying assignment to $\varphi$. This is known as the satisfiability problem, or SAT. SAT is seen
to be a problem that is at the center of computational complexity because it has been commonly used as a framework
to solve other combinatorial problems. In addition, SAT has found many uses in practice as well. Some of these 
examples include : AI-planning, software model checking, etc\cite{Silva08}.

Because of its importance, many other variants of the satisfiability problem have also been explored. 
One such important variant is the Exact Satisfiability problem, XSAT,
where it asks if one can find a satisfying assignment such that exactly one of the literal in each clause is assigned 
the value ``1" and all other literals in the same clause are assigned ``0". All the mentioned problems, 
SAT and XSAT, are both known to be NP-complete \cite{Sch78,Cook71}.

In this paper, we will focus on the XSAT problem and in particular, exact algorithms to solve it. 
XSAT is a well-studied problem and has seen numerous improvements \cite{SS81,MSV81,BMS05,D06} to it, with the
fastest solving it in $O(1.1730^n)$ time. 

In this paper, we will propose an algorithm to solve XSAT in $O(1.1674^n)$ time, using polynomial space. 
Like most of the earlier authors, we will design a Davis-Putnam-Logemann-Loveland 
(DPLL) \cite{DPLL60} style algorithm to solve this problem. We build our work upon the works of the earlier authors. 
While the earlier authors all used the standard measure, which is the number of variables $n$, we propose the
use of a nonstandard measure to help us to tighten the analysis of the algorithm further. 

\section{Preliminaries}

In this section, we will introduce some definitions and also the techniques needed
to understand the analysis of DPLL algorithm. 

\subsection{Branching factor and vector} 
Our algorithm is a DPLL style algorithm, or also known as a branch and bound algorithm. 
DPLL algorithms are recursive in nature and have two kinds of rules associated with them : Simplification and Branching rules.
Simplification rules help us to simplify a problem instance or to act as a case to terminate the algorithm.
Branching rules on the other hand, help us to solve a problem instance by recursively solving smaller instances of the problem. 
To help us to better understand the execution of a DPLL algorithm, the notion of a search tree is commonly used. 
We can assign the root node of the search tree to be the original problem, while subsequent child nodes are assigned 
to be the smaller instances of the problem whenever we invoke a branching rule. For more information of this area,
one may refer to the textbook written by Fomin and Kratsch \cite{FK10}. 

Let $\mu$ be our measure of complexity. To analyse the running time of DPLL algorithms, one just needs to bound
the number of leaves generated in the search tree. This is because the complexity of such
algorithm is proportional to the number of leaves, modulo polynomial factors, that is, 
$O(poly(|\varphi|,\mu) \times \text{number of leaves in the search tree})=O^*(\text{number of leaves in the search tree})$,
where the function $poly(|\varphi|,\mu)$ is some polynomial dependent on $|\varphi|$ and $\mu$, and $O^*(f(\mu))$ is the
class of all function $g$ bounded by some polynomial $p(\cdot)\times f(\mu)$. 

Then we let $T(\mu)$ denote the maximum number of leaf nodes generated by the algorithm when we have $\mu$ as
the parameter for the input problem.
Since the search tree is only generated by applying a branching rule, 
it suffices to consider the number of leaf nodes generated by that rule (as simplification rules take only polynomial time).
To do this, we use techniques in \cite{Kul99}. Suppose a branching rule has $r \geq 2$ children, with
$t_1,t_2 ,\ldots,t_r$ decrease in measure for these children.
Then, any function $T(\mu)$ which satisfies $T(\mu) \geq T(\mu-t_1) + T(\mu-t_2) + \ldots T(\mu-t_r)$, with appropriate
base cases, would satisfy the bounds for the branching rule. To solve the above linear recurrence, 
one can model this as $x^{-t_1} + x^{-t_2} + \ldots + x^{-t_r} = 1$. Let $\beta$ be the unique positive root
of this recurrence, where $\beta \geq 1$. Then any $T(\mu) \geq \beta^\mu$ would satisfy the recurrence for
this branching rule. In addition, we denote the branching factor $\tau(t_1,t_2,\ldots,t_r)$
as $\beta$. 
If there are $k$ branching rules in the DPLL algorithm, then the overall
complexity of the algorithm is the largest branching factor among all $k$ branching rules;
i.e. $c=max\{\beta_1,\beta_2,\ldots,\beta_k\}$, and therefore the time complexity of the
algorithm is bounded above by $O^*(c^\mu)$.

Next, we will introduce some known results about branching factors. If $k < k'$, then we have that
$\tau(k',j) < \tau(k,j)$, for all positive $k,j$. In other words, comparing two branching factor,
if one eliminates more weights, then this will result in a a smaller branching factor. Suppose that 
$i+j = 2\alpha$, for some $\alpha$, then $\tau(\alpha,\alpha) \leq \tau(i,j)$. In other words, 
a more balanced tree will result in a smaller branching factor. 

Finally, the correctness of DPLL algorithms usually follows from the fact that all cases have been covered. 

\subsection{Definitions}

\begin{definition}
A clause is a disjunction of literals. We also say that a clause is a multiset of literals. 
A $k$-literal clause is a clause $C$ with $|C|=k$. Let $C$ be a clause, then $\delta$ is a subclause
of $C$ if $\delta \subset C$.
\end{definition}

Suppose we have $C=(a \vee b \vee c \vee d)$, then $C$ is a $4$-literal clause. In addition, 
$\delta=(a \vee b \vee c)$ is a subclause of $C$. We may also write $C = (\delta \vee d)$. For now,
we define a clause as a multiset of literals as the same literal may appear twice in a clause. When no simplification rules
\footnote{More details later in Section 3, when the algorithm is given} 
can be applied, we may then think of a clause as a set of literals instead. 

\begin{definition}
Two clauses are called neighbours if they share at least a common variable. Two variables are called neighbours
if they appear in some clause together. Let $C_1$ and $C_2$ be two clauses that are neighbours.
Now if $|C_1\cap C_2|=k \geq2$, we say that $C_1$ and $C_2$ have $k$ overlapping variables. In addition, 
the variables in $C_1-C_2$ and $C_2-C_1$ are known as outside variables. Let $|C_1-C_2|=i$ and $|C_2-C_1|=j$,
$i,j\geq1$. Then we say that there are $i+j$ outside variables, in an $i$-$j$ orientation.  
\end{definition}

Note that this definition ($i$-$j$ orientation) is strictly used for the case when we have 
$k\geq2$ overlapping variables between any two clauses \footnote{Mainly in Section 4.3}.
 We only consider $i,j\geq1$ because if $i$ or $j$ is $0$, then
one of the clause must be a subclause of the other. Consider the following example.

\begin{example} \label{example_overlap}
Let $C_1 = (a \vee b \vee c \vee d \vee e)$ and $C_2=(d \vee e \vee f \vee g \vee h)$. Then in this case,
since $C_1 \cap C_2 = \{d,e\}$, there are 2 literals in the intersection and we say that $C_1$ and $C_2$ have
$2$ overlapping variables. In addition, $C_1 - C_2 = \{a,b,c\}$ and $C_2 - C_1 = \{f,g,h\}$. Now, we say
$C_1$ and $C_2$ have 6 outside variables in a $3$-$3$ orientation. 
\end{example} 

\begin{definition}
Let $x$ be a literal. Now the degree of a variable, $deg(x)$, denotes the total number of times that the literal $x$
and $\neg x$ appears in $\varphi$. If $deg(x)\geq3$, then we say that the variable $x$ is heavy . 
Further, for a heavy variable $x$ that appears in clauses $C_1,C_2,...,C_k$, $k\geq3$, we say that $x$ is in
$(l_1,l_2,...,l_k)$, where $|C_i|=l_i$, $1\leq i \leq k$. Adding on to this, 
\begin{enumerate}
\item if $\neg x$ appears in $C_i$, then we say $x$ is in $(l_1,l_2,...,\neg l_i,...,l_k)$.
\item if $|C_i|\geq l_i$, then we say $x$ is in $(l_1,l_2,...,\geq l_i,...,l_k)$.
\end{enumerate}
\end{definition}

Note that if $x$ is a heavy variable, we will only use this definition that $x$ is in $(l_1,l_2,...,l_k)$, whenever given any two
clauses that $x$ is in, they have at most 1 overlapping variable between them. 

\begin{example}
Suppose we have the following clauses : $(x \vee a \vee b \vee c \vee d)$, $(\neg x \vee e \vee f \vee g)$
, $(x \vee h \vee i \vee j \vee k)$. Then in this case, we have $x$ in $(5,\neg4,5)$.
We can also say that $x$ is in $(\geq4,\neg4,5)$ and we use ``$\geq i$" whenever we just need to know that the clause
length is at least $i$. Note that the order in which the clause length is presented here does not matter, i.e. 
 $(5,\neg4,5)$ can also be written as $(\neg 4,5,5)$.
\end{example}

\begin{definition}
We say that two variables, $x$ and $y$, are linked when we can deduce either $x=y$ or $x=\neg y$. When this happens,
we can proceed to remove one of the linked variable, either $x$ or $y$, and replace by the other.
\end{definition}

Suppose we have a 3-literal clause $(0 \vee x \vee y)$, by definition of being exact satisfiable, we can deduce that 
$x=\neg y$ in this case, and proceed to remove one variable, say $x$, by replacing 
all instances of $x$ by $\neg y$ and $\neg x$ by $y$ respectively. 

\begin{definition} \label{def_branching}
Given a formula $\varphi$ and $\delta$ a multiset of literals. 
\begin{enumerate}
\item If $|\delta|=1$, then let $x$ be the only literal in $\delta$. Now $\varphi[x=1]$ and $\varphi[x=0]$ denotes
the new formula obtained after assigning $x=1$ and $x=0$ respectively.
\item If $|\delta|\geq2$, then we only allow the following when $\delta \subset C$, for some clause $C$ in $\varphi$. 
$\varphi[\delta=1]$ denotes the new formula obtained after assigning all the $C-\delta$ to be 0. By definition of being 
exact-satisfiable, this is saying that the ``1" must only appear in one of the literals in $\delta$. Therefore, all the 
literals in $C-\delta$ are assigned 0. On the other hand, $\varphi[\delta=0]$ denotes the new formula obtained 
after assigning all the literals in $\delta$ to be 0.
\end{enumerate}
Similarly, given two literals $x$ and $y$, we say that $\varphi[x=y]$ is the new formula obtained by replacing all
occurrences of $x$ by $y$. 
\end{definition}

\begin{example}
Suppose $\varphi=(a \vee b \vee c \vee d)$ and $\delta=(a \vee b \vee c)$. Then
$\varphi[\delta=1]=(a \vee b \vee c \vee 0)$ since we are saying that the ``1" appears in either $a$, $b$, or $c$.
On the other hand, $\varphi[\delta=0]=(0 \vee 0 \vee 0 \vee d)$.
\end{example}

Definition \ref{def_branching}.1 is used whenever we are branching a variable. 
On the other hand, Definition \ref{def_branching}.2 is used when we 
want to branch a subclause, especially when we deal with $k\geq2$ 
overlapping variables between two clauses. 
In addition, when we have a subclause 
$\delta$ such that $|\delta|=2$, then let $x$ and $y$ be the literals in $\delta$. Saying that 
$\varphi[\delta=1]$ is the same as saying $\varphi[x=\neg y]$, linking $x=\neg y$.

A common technique used by the earlier authors is known as resolution. If there are clauses 
$C_1 = (C \vee x)$ and $C_2 = (C' \vee \neg x)$, where $x$ is a literal, $C$ and $C'$ 
are subclauses of $C_1$ and $C_2$ respectively, then we can replace every clause
$(x \vee \alpha)$ by $(C' \vee \alpha)$, and every clause $(\neg x \vee \beta)$ by
$(C \vee \beta)$, for some subclause $\alpha,\beta$. In addition,
every literal in $C \cap C'$ can be assigned 0. This can help us to remove
literals appearing as $x$ and $\neg x$ in different clauses.

\subsection{A nonstandard measure}
Instead of using the number of variables as our measure, we will design a 
nonstandard measure to help us to improve the worst case time complexity of our algorithm. 
Let $\{x_1,x_2,...,x_n\}$ be the set of variables in $\varphi$. For $1\leq i \leq n$, 
we define the weight $w_i$ for $x_i$ as :
\[
 w_i = 
\begin{cases}
    0.8823,& \text{if $x_i$ is on a 3-literal clause such that all 3 variables in that }\\
 	    & \text{clause do not have the same neighbour}\\
    1,              & \text{otherwise}
\end{cases}
\]

We then define our choice of measure as $\mu = \sum_i w_i$, where $\mu \leq n$ by definition.
This value of $0.8823$ is chosen by a linear search program 
to bring down the overall runtime of the algorithm to as low as possible. Therefore,
we have $O(c^{\mu}) \subseteq O(c^n)$, for some constant $c\geq1$ by definition.

\begin{example}
Suppose we have the following clauses : $(x \vee y \vee z \vee a),(x \vee u \vee w \vee v),
(x \vee r \vee s \vee t),(a \vee v \vee t)$ and the clause $(y \vee e \vee f)$.
The variables $x,z,u,w,r$ and $s$ have weight 1. By definition, variables $a,v$ and $t$ are
assigned the weight 1 because these variables 
have $x$ as their neighbour. Variables $y,e,f$ have weights 
$0.8823$ because these 3 variables do not have the same neighbour.
\end{example}

\section{Algorithm}

All of our simplification rules and branching rules are designed 
to ensure that the overall measure does not increase after applying them. That is,
the measure before applying any of the rule, $\mu$, and the measure after applying 
any of the rule, $\mu'$, is always $\mu' \leq \mu$. 
We call our DPLL algorithm $XSAT(.)$. Note that if every variable $x$ has $deg(x)\leq2$, 
then we can solve XSAT in polynomial time \cite{MSV81}. With this in mind, we'll design our algorithm
by branching all heavy variables. Note that each line of the algorithm
has decreasing priority; Line 1 has higher priority than Line 2, Line 2 than Line 3 etc.
Let $\alpha, \beta, \delta$ be subclauses.

\noindent
Algorithm : $XSAT$ \\
Input : A formula $\varphi$ \\
Output : 1 if $\varphi$ is exact satisfiable, else 0 \\

\begin{enumerate} 
\item If there is a clause that is not exact-satisfiable, then return 0.
\item If there is a clause $C=(1 \vee \delta)$ or $C=(x \vee \neg x \vee \delta)$, for some variable $x$, 
then set all literals in $\delta$ to 0 and drop the clause $C$. Return $XSAT(\varphi[\delta=0])$.
\item If there exist a clause $C=(0 \vee \delta)$, then update $C=\delta$. Update $\varphi'$ as the new formula and
return $XSAT(\varphi')$.
\item If there exist a 1-literal clause containing the literal $l$, then drop that clause. 
Return $XSAT(\varphi[l=1])$.
\item If there exist a 2-literal clause containing the literal $l$ and $l'$, then drop that 
clause. Return $XSAT(\varphi[l=\neg l'])$.
\item If there exist a clause $C$ with a literal $l$ appearing at least twice, then return $XSAT(\varphi[l=0])$.
\item If there exist clauses of the type $(\alpha \vee x \vee y)$ and $(\beta \vee x \vee \neg y)$, for some
literal $x$ and $y$, then return $XSAT(\varphi[x=0])$.
\item If there exist clauses of the type $(\alpha \vee x \vee y)$ and $(\beta \vee \neg x \vee \neg y)$, then
return $XSAT(\varphi[x=\neg y])$.
\item If there are clauses $C$ and $C'$ such that $C \subset C'$, then set all literals in $\delta=C'-C$ as 0, remove the clause $C'$
and return $XSAT(\varphi[\delta=0])$.
\item If there is a variable $x$ appearing in at least three 3-literal clauses, then we either simplify it or branch $x$.
If we simplify it, let $\varphi'$ be the new formula after simplifying. Return $XSAT(\varphi')$. If we branch $x$,
return $XSAT(\varphi[x=1]) \vee XSAT(\varphi[x=0])$.
\item If there are clauses $C_1$ containing $x$ and $C_2$ containing $\neg x$, for some literal $x$. Then we apply resolution and
let $\varphi'$ be the new formula. Return $XSAT(\varphi')$. 
\item If there are clauses $C_1$ and $C_2$ such that they have $k\geq2$ overlapping variables, then check if the outside variables
are in a $1$-$j$ orientation, $j\geq1$. If yes, then let $\varphi'$ be the new formula after applying some changes 
\footnote{Full details given in the Section \ref{sectionline12}}, then return $XSAT(\varphi')$. Else, let $\delta = C_1 \cap C_2$ 
and we branch the subclause $\delta$. Return $XSAT(\varphi[\delta=1]) \vee XSAT(\varphi[\delta=0])$.
\item If there is a heavy variable $x$, then branch $x$. Return $XSAT(\varphi[x=1]) \vee XSAT(\varphi[x=0])$.
\item If all the variables $x$ have $deg(x)\leq2$, then solve the problem in polynomial time. Return 1 if exact-satisfiable, 
else return 0. 
\end{enumerate}

Lines 1 to 9, 11 are simplification rules, while Lines 10, 12 and 13 are branching rules. Line 14 takes only polynomial time
to decide if there is an exact-satisfiable assignment to $\varphi$ when $deg(x)\leq2$ for all variable $x$. 
Line 1 says that if any clause is found not to be exact-satisfiable, then we can return 0. 
Line 2 says if a clause contains a ``1", then the other literals appearing in the clause must be assigned 0. 
Line 3 says that if we have a clause containing ``0", then we can update that clause by dropping off the constant ``0". 
Line 4 says that if we encounter a 1-literal clause, then that literal must be assigned 1. 
Line 5 says that if there are any 2-literal clause containing some literals $x$ and $y$, then we can just link the
two literals $x=\neg y$ together. After Line 5 of the algorithm, every clause in $\varphi$ must be at least a 3-literal clause. 

Line 6 deals with clauses containing the same literals that appear at least twice. After Line 6, every clause can only contain
any literal at most once. Lines 7 and 8 deals with two clauses that have at least two variables in common, in different permutations. 
After Line 8, if any two clauses have at least two variables in common, then this implies that they have share at least two
literals in common. After Line 9, no clause is a subclause of a larger clause in $\varphi$.

In Line 10, we deal with variables that appears in at least three 3-literal clauses. We deal with this case early on
because it helps us to reduce the number of cases that we need to handle later on while branching in Section 4.3 and 4.4.
In Line 11, we deal with clauses $C_1$ containing the literal $x$ and $C_2$ containing $\neg x$. 
Line 12 deals with two clauses having $k\geq2$ overlapping variables. First, we deal with such cases
in a $1$-$j$ orientation, $j\geq1$, followed by such cases in an $i$-$j$ orientation, $i,j\geq2$. 
After which, any two clauses must have only at most one variable in common.
Line 13 deals with heavy variables. After that, no heavy variables exist in the formula $\varphi$
and we can proceed to solve the problem in polynomial time in Line 14. We have therefore covered all cases in our algorithm.

\section{Analysis of Algorithm}

In this section, we will analyze the overall runtime of the algorithm given in the previous section. 
Note that simplification rules only take polynomial time. Therefore, we will
analyse from Lines 10 to 13 of the algorithm. 

Due to the way we design our measure, if a $k$-literal clause drops to a 3-literal clause, $k>3$, 
we can factor in the change of measure 
of $1-0.8823=0.1177$ for each of the variables in the 3-literal clause, if there is no common neighbour. 
Whenever we are dealing with a 3-literal clause, for simplicity, we will treat all the variables in it as having
a weight of $0.8823$ instead of $1$. This gives us an upper bound on the branching factor without the need to
consider all kinds of cases.

In addition, when we are dealing with 3-literal clause, sometimes we have to increase the measure after linking.
For example, suppose we have the clause $(0 \vee x \vee y)$, for some literals $x$ and $y$. Now we can link
$x =\neg y$ and proceed to remove one variable, say $x$. This means that the 3-literal clause is removed and the
surviving variable $y$, may no longer be appearing in any other 3-literal clause. Therefore, the weight of $y$ increases
from $0.8823$ to $1$. This increase in weight means that we increase our measure and therefore, we have to factor in
``-0.1177" whenever we are linking variables in a 3-literal clause together.

\subsection{Line 10 of the Algorithm} \label{sectionline10}

Line 10 of the algorithm deals with a variable appearing in at least three 3-literal clauses. We can either simplify the case further, or
branch $x$. At this point in time, Lines 11 and 12 of the algorithm has not been called. 
This means that we have to deal with literals appearing as $x$ and $\neg x$, and that given any two clause,
it is possible that they have $k\geq2$ overlapping variables. 

\begin{lemma}
If $x$ appears in at least three 3-literal clauses, we either simplify this case further or we branch $x$, 
incurring at most $O(1.1664^n)$ time.
\end{lemma}
\begin{proof}
Now let $x$ be appearing in two 3-literal clauses. We first deal with the case that that for any two 3-literal clauses,
there are $k\geq2$ overlapping variables. Since simplification rules do not apply anymore, the only case we need
to handle here is $(x \vee y \vee z)$ and $(x \vee y \vee w)$, for some literals $w,y,z$.
 In this case, we can link $w=z$ and drop one of these clauses.

For the remaining cases, $x$ must appear in three 3-literal clause and there are no $k\geq2$ overlapping variables
between any two of the 3-literal clause. Therefore, for the remaining case,
$x$ must be in $(3,3,3)$ or $(3,3,\neg 3)$.

For the $(3,3,3)$ case, let the clauses be $(x \vee v_1 \vee v_2)$, $(x \vee v_3 \vee v_4)$ and 
$(x \vee v_5 \vee v_6)$, where $v_1,...,v_6$ are unique literals. We branch $x=1$ and $x=0$ here.
When $x=1$, we remove the variables $v_1,...,v_6$ and $x$ itself. This gives us a change of measure of
$7\times0.8823$. When $x=0$, we remove $x$, and link $v_1 =\neg v_2$, $v_3 =\neg v_4$ and 
$v_5 =\neg v_6$. This gives us a change of measure of $4\times0.8823-3\times0.1177$. This givs us a branching
factor of $\tau(7\times0.8823,4\times0.8823-3\times0.1177)=1.1664$. The case for $(\neg 3,\neg 3,\neg 3)$
is symmetric.

For the $(3,3,\neg 3)$ case, let the clauses be $(x \vee v_1 \vee v_2)$, $(x \vee v_3 \vee v_4)$ and 
$(\neg x \vee v_5 \vee v_6)$, where $v_1,...,v_6$ are unique literals. Again, we branch $x=1$ and $x=0$.
When $x=1$, we remove $x$ and the variables $v_1,...,v_4$, and link the variables $v_5 =\neg v_6$. This gives us
a change of measure of $6\times0.8823-0.1177$. When $x=0$, we remove $x$, $v_5,v_6$, and link the variables
$v_1 = \neg v_2$ and $v_3 =\neg v_4$. This gives us a change of measure of $5\times0.8823-2\times0.1177$.
This gives us a branching factor of $\tau(6\times0.8823-0.1177,5\times0.8823-2\times0.1177)=1.1605$.
The case for $(3,\neg3,\neg3)$ is symmetric. Therefore, this takes at most $O(1.1664^n)$ time.
\end{proof}

\subsection{Line 11 of the Algorithm} \label{sectionline11}

Line 11 of the algorithm applies resolution. One may note that our measure is designed in terms of the length of the clause.
Therefore, it is possible that the measure may increase from $0.8823$ to $1$ after applying resolution. 
Applying resolution on $k$-literal clauses, $k\geq4$, 
is fine because doing so will not increase the measure. On the other hand, 
applying on 3-literal clauses will increase the length of the clause and hence, increase the weights of the
other variables in the clause, and finally, the overall measure.
Therefore to apply resolution on such cases, we have to ensure that the removal of the variable $x$, is more than the increase of the
weights of from $0.8823$ to $1$ $(1-0.8823=0.1177)$. 
To give an upper bound, we assume that $x$ has weight $0.8823$. Taking $0.8823\div 0.1177 = 7.5$.
Therefore, if there are more than 7.5 variables increasing from $0.8823$ to $1$, then we refrain from doing so. 
This translates to $x$ appearing in at least four 3-literal clauses. 
However, this has already been handled by Line 10 of the algorithm. 
Hence, when we come to Line 11 of the algorithm, we can safely apply resolution.

\subsection{Line 12 of the algorithm} \label{sectionline12}

In this section, we deal with Line 12 of the algorithm. Since simplification rules do not apply anymore when
this line is reached, we may then think of clauses as sets (instead of multiset) of literals, since the same literal
can no longer appear more than once in the clause. In addition, from the previous line of the algorithm,
we know that we will not have $x$ and $\neg x$ appearing in the formula, for any literal $x$. 
Now, we fix the following notation
for the rest of this section. 
Let $C_1$ and $C_2$ be any clauses given such that $C_1 \cap C_2=\delta$, with $|\delta|\geq 2$ overlapping variables, in an
$i$-$j$ orientation, where $|C_1 - C_2|=i$ and $|C_2 - C_1|=j$, where $i,j \geq1$. We divide them into 3 parts, 
let $L = C_1 - C_2$ (left), $R = C_2 - C_1$ (right) and $\delta$ (middle). 
For example, in Example \ref{example_overlap}, we have $L=\{a,b,c\}$ and $R=\{f,g,h\}$. 
We first deal with the cases $i=1$, $j\geq1$.

\begin{lemma}
The time complexity of dealing with two clauses with $k\geq2$ overlapping variables, having $1$-$j$ orientation, $j\geq1$,
is at most $O(1.1664^n)$.
\end{lemma}
\begin{proof}
If $j=1$, then let $x \in L$ and $y \in R$. Then we can just link $x=y$
and this case is done. If $j \geq2$, then let $C_1 = (x \vee  \delta)$ and $C_2 = (\delta \vee R)$. 
From $C_1$, we know that $\neg x = \delta$. Therefore, $C_2$ can be rewritten has $(\neg x \vee R)$. With the
clauses $C_1 = (x \vee \delta)$ and $C_2 = (\neg x \vee R)$, we can apply Line 11 of the algorithm 
which either uses resolution to remove the literals $x$ and $\neg x$, or to apply branching to get a complexity of
$O(1.1664^n)$.   
\end{proof}

Now, we deal with the case of having $k\geq2$ overlapping variables in an $i$-$j$ orientation, $i,j\geq2$.
Note that during the course of branching $\delta=0$, when a longer clause drops to a 3-literal clause $L$ (or $R$), then we
can factor in the change of measure of $1-0.8823=0.1177$ for each of the variable in $L$ (Normal Case). 
However, there are situations when we are not allowed to factor in this change.
Firstly, when there is a common neighbour to the variables in $L$ (Case 1). Secondly, when some or all 
variables in $L$ already have weights $0.8823$, which means the variable appears in further 
3-literal clauses prior to the branching (Case 2).

Instead of enumerating every single case, we show that some cases can be avoided by upper bounding them
from a different case. We first show how to deal with Case 1. \\

\noindent\textbf{Case 1:} The variables in $L$, with $|L|=3$,  (similarly for $R$) have a common neighbour.
\begin{itemize}
\item When there is a clause $L'$, such that $L \subset L'$ and therefore every variable in $L'-L$ is a neighbour to $C$.
However, if this case happens. Then by our simplification rule, we can set the literals in $L'-L$ to 0. We can
remove at least one such variable, and the weight of such a variable is at least 0.8823.
\item Let the literals in $L$ be $a,b,c$, $\alpha,\beta,\gamma$ be subclauses.
\begin{enumerate}
\item $(s \vee \beta \vee a \vee b)$ and $(s \vee \alpha \vee c)$
\item $(s \vee \alpha \vee a)$, $(s \vee \beta \vee b)$ and $(s \vee \gamma \vee c)$ 
\end{enumerate}
Then in the above 2 cases, $s=0$ and the weight of $s$ is at least $0.8823$
\end{itemize}

In all 3 possible cases in Case 1, we are able to factor in an additional measure of $0.8823$. Now let $\Delta\mu_{\delta=1}$ 
be the change of measure when we branch $\delta=1$ and $\Delta\mu_{\delta=0}$ when we branch $\delta=0$ 
for the Normal Case. Note that when $\delta=0$, we remove all the variables in $\delta$ and we can also factor in the change of measure
for at most 6 variables (in a $3$-$3$ orientation). In Case 1, we can remove an additional variable that has weight at least $0.8823$,
which means $0.8823-6\times0.1177=0.1761$, allowing us to factor in additional change of measure of $0.1761$ in the worst case.
Therefore, we have $\tau(\Delta\mu_{\delta=1},\Delta\mu_{\delta=0}+0.1761) < \tau(\Delta\mu_{\delta=1},\Delta\mu_{\delta=0})$,
being upper bounded by the branching factor in the Normal Case. Hence, it suffices to just show the Normal Case.

For Case 2, we pay special attention to the outside variables in an $i$-$j$ orientation, $i\leq3$ or $j\leq3$. This is because when $i,j\geq4$,
and while branching $\delta=0$, we can only remove the variables in $\delta$ 
and not factor in other changes in measure from the variables in $L$ or $R$. 
On the other hand, when $\delta=1$, we can remove additional variables not in $L$, $R$ and $\delta$, 
whenever we have a variable having weight $0.8823$. Let $s$ be a variable not appearing in 
$L$, $R$ or $\delta$. We show all the possibilities below. 

\noindent\textbf{Case 2 :} The variables in $L$ (or $R$) appear in further 3-literal clauses.
\begin{enumerate}
\item Case 2.1 : A pair of 3-literal clauses containing $s$, with the neighbours of $s$ appearing in $L$ and $R$. 
For example, if we have $(l_1 \vee l_2 \vee  \delta)$ and $(\delta \vee r_2 \vee r_1)$,
and the two 3-literal clauses $(s \vee l_1 \vee r_1)$ and $(s \vee l_2 \vee r_2)$. 

In such a case, we branch $s=1$ and $s=0$.
Now when we branch $s=1$, we remove at least $s,l_1,r_1,l_2,r_2$. When $s=0$, we link $l_1=\neg r_1$ and $l_2 = \neg r_2$.
Then, the new clauses will be $(\neg r_1, \neg r_2 \vee \delta)$ and $(\delta \vee r_2 \vee r_1)$. Then, by our simplification rule,
we must have that $\neg r_1 = r_2$, and we can remove $\delta$. To upper bound this branching factor, we treat all the variables
as having weight $0.8823$. This gives us a branching factor of $\tau(5\times0.8823,(4+|\delta|)\times0.8823)$.
Since $|\delta|\geq2$, our branching factor is bounded above by $1.1541$.

\item Case 2.2 : Not Case 2.1. In other words, there is no such $s$ that appears in two 3-literal clauses, where the neighbours of $s$
are the variables in $L$ and $R$. In this case, we have 3-literal clauses, each containing a variable from $L$, a variable from $R$,
and another variable not from $L$, $R$, and $\delta$. 
\end{enumerate}

By Line 10 of the algorithm, $s$ cannot appear in a third 3-literal clause. Therefore, we must either have Case 2.1 or Case 2.2.

For Case 2.1, we have shown that if such a case arises, then $1.1541$ acts as an upper bound for all such cases of $k\geq2$ 
overlapping variables in an $i$-$j$ orientation, $i,j\geq2$. Therefore, in the Lemma below, we will not deal with such cases.

Case 2.2 arises when it is not Case 2.1; when there is no such $s$, appearing in two 3-literal clauses, with
the neighbours of $s$ appearing in $L$ and $R$. Case 2.2 represents the case where we can have $(l \vee s \vee r)$,
where $l \in L$, $r \in R$ ($s$ only appears in exactly one 3-literal case in Case 2.2). 

Note that, apart from such a scenario in Case 2.2, we can of course have
a variable appearing in a further 3-literal clause, containing a variable from $L$ or $R$, and then containing two
variables not from $L$, $R$ and $\delta$ (Standalone 3-literal). For example, 
we have $C_1 = (a \vee b \vee c \vee \delta)$ and $C_2 = (\delta \vee d \vee e \vee f)$. 
So Case 2.2 has 3-literal clauses like $(c \vee s \vee d)$. However,
we can also have Standalone 3-literal clauses like $(f \vee g \vee h)$, where $g,h$ does not appear in $L$, $R$
and $\delta$. 

We can show that Case 2.2 upper bounds the case of having Standalone 3-literal. 
Given any case of $k\geq2$ overlapping variables, in an $i$-$j$ orientation,
$i,j\geq2$, let our clauses be $(\alpha \vee x \vee \delta)$ and $(\delta \vee y \vee \beta)$, for some subclause
$\alpha,\beta$ and $\delta$. We will now compare Case 2.2 with the case of having Standalone 3-literal clauses.
We can have two Standalone 3-literal clause, on the variables $x$ and $y$, or we can have only one Standalone 
3-literal clause, on either $x$ or $y$. Now let $\Delta\mu_{\delta=1}$ ($\Delta\mu_{\delta=0}$) 
denote the change of measure for all the variables
except for $x$ and $y$ when we branch $\delta=1$ ($\delta=0$).
 
\begin{itemize}
\item We can have a single Standalone 3-literal clause $(x \vee v_1 \vee v_2)$, where $v_1,v_2$ is not from $\alpha$, $\beta$ and $\delta$.
Then this case gives us a branching factor of $\tau(\Delta\mu_{\delta=1}+1+0.8823+0.7646,\Delta\mu_{\delta=0})
\leq \tau(\Delta\mu_{\delta=1}+3\times0.8823,\Delta\mu_{\delta=0})$ (Case 2.2). We have $1$ from the removal of $y$,
$0.8823$ from $x$ and $0.7646$ $(0.8823-0.1177)$ from linking $v_1$ and $v_2$. 
\item We can have a clause $(x \vee v_1 \vee v_2)$ and $(y \vee v_3 \vee v_4)$, where $v_1,...,v_4$ are not from
$\alpha$,$\beta$, $\delta$. Then this gives us $\tau(\Delta\mu_{\delta=1}+2\times(0.8823+0.7646),\Delta\mu_{\delta=0})
\leq \tau(\Delta\mu_{\delta=1}+3\times0.8823,\Delta\mu_{\delta=0})$ (Case 2.2).
\end{itemize}

Therefore, we see that the branching factor in Case 2.2 acts as an upper bound for the Standalone 3-literal case. 
Finally, we show that we can just treat all the variables in $\delta$ as having weight $1$ instead of $0.8823$.
Suppose a variable in $\delta$ appears in a 3-literal clause. Then the same 3-literal clause cannot contain another variable
from $L,R$ or $\delta$ it would be a $1$-$j$ orientation that would have already be handled earlier. So this 3-literal clause
must be a Standalone. Let $\Delta \mu_{\delta=1}$
be the change of measure when branching $\delta=1$ and $\Delta \mu_{\delta=0}$ be the change of measure when
branching $\delta=0$ for any case when the weight of variables in $\delta$ is 1. 
For $|\delta|\geq3$, the variable in $\delta$ appears in a 3-literal clause, then the branching will give us 
$\tau(\Delta \mu_{\delta=1},\Delta \mu_{\delta=0}+0.6469) < \tau(\Delta \mu_{\delta=1},\Delta \mu_{\delta=0})$. 
Note that when we have a Standalone 3-literal clause, we have a change of measure of $0.8823+0.7646$ when $\delta=0$. 
Now the difference between this and when the weight is 1 is $0.8823+0.7646-1=0.6469$. 
When $|\delta|=2$, we apply linking when we branch $\delta=1$. This gives us 
$\tau(\Delta \mu_{\delta=1}+0.8823,\Delta \mu_{\delta=0}+0.6469) < \tau(\Delta \mu_{\delta=1}+1,\Delta \mu_{\delta=0})$.
Therefore, we will always get a better branching factor because the search tree becomes more balanced.
Therefore it suffices to just deal with the case that the variables in $\delta$ have
weight 1 for our analysis below.

For the Lemma below, we will only show the Normal Case and Case 2.2 since these two cases upper bounds
the other cases as shown above. 

\begin{lemma}
The time complexity of dealing with two clauses with $k\geq2$ overlapping variables, having $i$-$j$ orientation,
$i,j\geq2$, is at most $O(1.1674^n)$ time.
\end{lemma}
\begin{proof}
Let any two clauses be given with $k\geq2$ overlapping variables and have at least 4 outside variables in a $2$-$2$ orientation.
We will show the Normal Case first, followed by Case 2.2 (only for outsides variables $i\leq3$ or $j \leq3$). For Case 2.2,
and the appearance of each 3-literal clause, note that when branching $\delta=1$, we can remove all the variables in the 3-literal
clause, giving us $3\times0.8823$ per 3-literal clause that appears in this manner. Let $h$ denote the number of further 3-literal
clauses for Case 2.2 encountered below. In addition, for Case 2.2 having odd number of outside variables,
we treat the variable not in any 3-literal clause as having weight 1, acting as an upper bound to our cases.

For $k=2$, and we have 4 outside variables in a $2$-$2$ orientation. 
When $\delta=1$, we remove all 4 outside variables and another 1 from linking the variables in $\delta$. This gives us a change of measure 5.
When $\delta=0$, we remove 2 variables in $\delta$ and another 2 from linking the variables in $L$ and $R$. This gives us a 
change of measure of 4. Therefore, we have a branching factor of $\tau(5,4)=1.1674$. 
For Case 2.2, we can have at most two 3-literal clauses here. This gives us 
$\tau(h\times(3\times0.8823)+2\times(2-h)+1,2+2\times0.8823)$, 
$h\in \{1,2\}$, which is at max branching factor of
$1.1612$, when $h=1$. This completes the case for 4 outside variables.

Suppose we have 5 outside variables in $2$-$3$ orientation. Branching $\delta=1$ will remove all outside variables, 
and 1 of the linked variable in $\delta$. This gives us a change of measure of 6. On the other hand, branching $\delta=0$
will allow us to remove all the variables in $\delta$, link the 2 variables in $L$, and factor in the change of measure
for the remainining variables in $R$. This gives us a change of measure of $\tau(6,3+3\times0.1177)=1.1648$.
For Case 2.2, we have at most two 3-literal clauses appearing in both $L$ and $R$. 
Then we have $\tau(h\times(3\times0.8823)+2\times(2-h)+2,2+0.8823+0.1177)$, $h\in\{1,2\}$,
which is at max branching factor of $1.1636$ when $h=1$. This completes the case for 5 outside variables.

Suppose we have 6 outside variables in a $3$-$3$ orientation. Branching $\delta=1$ will remove all 6 outside variables
in $L$ and $R$, and also remove an additional variable by linking the two variables in $\delta$. This gives us a change 
of measure of $7$. On the other hand, when $\delta=0$, we remove all the variables in $\delta$ and also
factor in the change of measure for the variables in $L$ and $R$, a total of $2+6\times0.1177$ for
this branch. This gives us a branching factor of $\tau(7,2+6\times0.1177)=1.1664$. 
When Case 2.2 applies, then we can have at most three 3-literals clauses appearing. This gives us a
branching factor of $\tau(h\times(3\times0.8823)+2\times(3-h)+1,2+2\times(3-h)\times0.1177)$, $h\in \{1,2,3\}$, 
with max branching of $1.1641$ when $h=1$. This completes the case for 6 outside variables.

Suppose we have 7 outside variables in a $3$-$4$ orientation. Branching $\delta=1$ will allow
us to remove all 7 outside variables, and 1 variable from $\delta$ via linking. This gives us a change of measure of 8.
On the other hand, when $\delta=0$, we can factor in a change of measure of $3\times0.1177$ from the variables.
This gives us a branching factor of $\tau(8,2+3\times0.1177)=1.1630$.  
For Case 2.2, there are at most three 3-literal clauses between
$L$ and $R$. This gives us a branching factor of 
$\tau(h\times(3\times0.8823)+2\times(3-h)+1+1,2+(3-h)\times0.1177)$,
which is at max of $1.1585$ when $h=1$. This completes the case for
7 outside variables.

Let $p\geq8$ be the number of outside variables. Branching $\delta=1$ allows us to remove all $p$ outside variables,
and an additional variable from linking in $\delta$, which has a change of measure of 9. 
For the $\delta=0$ branch, we remove two variables. This gives us a branching factor of 
$\tau(p+1,2)\leq \tau(9,2)=1.1619$. This completes the case for $k=2$ overlapping variables. 

Now we deal with $k=3$ overlapping variables. If there are 4 outside variables in a $2$-$2$ orientation, then
branching $\delta=1$ will allow us to remove all 4 outside variables, which is a change of measure of $4$. 
On the other hand, branching $\delta=0$ will allow us to
remove all the variables in $\delta$, as well as link the two variables in $L$ and $R$, removing
a total of 5 variables. This gives $\tau(4,5)=1.1674$. When we have Case 2.2, then we have at most
two 3-literal clauses appearing. This gives us a branching factor of at most
$\tau(h\times(3\times0.8823)+2\times(2-h),3+2\times0.8823)$, $h \in \{1,2\}$, which has a max
branching factor of $1.1588$ when $h=1$. This completes the case for 4 outside variables in a $2$-$2$ 
orientation.

For the case of 5 outside variables, they are in a $2$-$3$ orientation. Branching $\delta=1$
will allow us to remove all 5 outside variables. On the other hand, branching $\delta=0$ will allow us to
remove all the variables in $\delta$, an additional variable from 
linking the two variables in $L$, as well as factoring in the change of measure from $R$ of
$4+3\times0.1177$. This gives us a branching factor of
$\tau(5,4+3\times0.1177)=1.1601$. For Case 2.2, we can have at most two 3-literal clauses occurring. 
Then we have a branching factor of $\tau(h\times(3\times0.8823)+2\times(2-h)+1,3+0.8823+0.1177)$, $h \in \{1,2\}$, 
which is at max branching factor of $1.1563$ when $h=1$. This completes the case for 5 outside variables.

For the case of 6 outside variables, they are in a $3$-$3$ orientation. When branching $\delta=1$,
we can remove all 6 outside variables. When branching $\delta=0$, we remove all 3 variables in $\delta$,
and we can factor in the change of measure for these of $0.1177$ for these 6 variables. This gives a
branching factor of $\tau(6,3+6\times0.1177)=1.1569$. In Case 2.2, we can have at most three 3-literal
appearing in $L$ and $R$. Then the branching factor for this case would be 
$\tau(h\times(3\times0.8823)+2\times(3-h),3+2\times(3-h)\times0.1177)$,  $h\in \{1,2,3\}$,
which is at max branching factor of $1.1526$ when $h=1$. This completes the case for 6 outside variables. 

Let $p\geq7$ be the number of outside variables. Then branching $\delta=1$ will allow us to remove at
least 7 variables, and when $\delta=0$, we remove all the variables in $\delta$. This gives us a branching
factor of $\tau(p,3) \leq \tau(7,3)=1.1586$. For Case 2.2, we can have at most $h\leq \floor{\frac{7}{2}}$ 
3-literals clauses. Then our branching factor is 
$\tau(h\times(3\times0.8823)+2\times(\floor{\frac{7}{2}}-h)+1,3)$, which is at max of $1.1503$ when $h=1$.
This completes the case for $k=3$ overlapping variables.

Now, we deal with the case of $k=4$ overlapping variables. When we have 4 outside variables
in a $2$-$2$ orientation, then branching $\delta=1$ will allow us to remove all 4 outside variables,
giving us a change of measure of 4. On the other hand, when $\delta=0$, we remove all 
the variables in $\delta$, and link the two variables in $L$ and $R$. This gives us a change of measure
of 6. Therefore, we have a branching factor of $\tau(4,6)=1.1510$. For Case 2.2,
we can have at most two 3-literal clauses. Then our branching factor is 
$\tau(h\times(3\times0.8823)+2\times(2-h),4+2\times0.8823)$, $h \in \{1,2\}$, which is at max 
$1.1431$ when $h=1$. This completes the case for 4 outside variables. 

If there are $p\geq5$ outside variables,
then branching $\delta=1$ will allow us to remove at least 5 variables. On the other hand,
branching $\delta=0$ will remove all the variables in $\delta$. This gives a branching factor of
$\tau(p,4) \leq \tau(5,4)=1.1674$. For Case 2.2, we can have at most $h\leq\floor{\frac{5}{2}}$
number of 3-literal clauses. 
The branching factor is $\tau(h\times(3\times0.8823)+2\times(\floor{\frac{5}{2}}-h)+1,4)$, 
which is at max of $1.1563$ when $h=1$. This completes the case for 5 outside variables.

Finally for $k\geq5$ overlapping variables and $p\geq4$ outside variables, branching $\delta=1$ will remove at least 
4 variables, while branching $\delta=0$ will remove at least 5 variables. This gives us 
$\tau(p,k) \leq \tau(4,5) = 1.1674$. For Case 2.2, there can be at most $h\leq \floor{\frac{4}{2}}$
number of 3-literal clauses. Then our branching factor is at most 
$\tau(h\times(3\times0.8823)+2\times(\floor{\frac{p}{2}})-h),k) \leq 
\tau(h\times(3\times0.8823)+2\times(\floor{\frac{4}{2}}-h),5)$, which has max franching factor of $1.1547$
when $h=1$.
 
This completes the case for $k\geq2$ overlapping variables
and the max branching factor while executing this line of the algorithm is $1.1674$.
\end{proof}

\subsection{Line 13 of the algorithm}
Now, we deal with Line 13 of the algorithm, to branch off heavy variables in the formula.
After Line 12 of the algorithm, given any two clauses $C_1$ and $C_2$, there can only be at most
only 1 variable appearing in them. Cases 1 and 2 in the previous section will also apply here.
In Section 4.3, we paid special attention to $L$ and $R$ when $|L|=3$ or $|R|=3$. Here,
we pay special attention to $x$ being in 4-literal clauses, because after branching $x=0$, it will
drop to a 3-literal clause. Since we have dealt with $(3,3,3)$ case earlier, here, we'll deal
with the remaining cases; cases from $(3,3,\geq4)$ to $(\geq5,\geq5,\geq5)$. 

For Case 1 (common neighbour), we will only show the analysis
 for the $(4,4,4)$ case because it is only this case where we can factor in
a change of $9\times0.1177 > 0.8823$, which is better 
than removing the common neighbour. For Case 2, there are some changes as well. 
Here, we are dealing with 3 clauses instead of 2 in the previous section. Therefore, there will be 
more permutation of 3-literal clauses to consider. Recall previously that we dealt with a case
where $s$ appears in two 3-literal clauses in Case 2.1 of Section 4.3. 
Here, we deal with something similar.

Suppose there are clauses $C_1 = (l_1 \vee l_2 \vee \delta \vee x)$ , $|C_1|\geq3$, 
$C_2 = (r_1 \vee r_2 \vee \alpha \vee x)$, $|C_2|\geq4$, for some subclause $\delta$ and $\alpha$,
$(s \vee l_1 \vee r_1)$ and  $(s \vee l_2 \vee r_2)$, $s$ not appearing in the clauses $C_1$ and $C_2$.
 We give an upper bound for such cases. When $s=1$,
we remove 5 variables, with a change of measure of $5\times0.8823$. When $s=0$, we remove $s$,
and link $l_1 = \neg r_1$, and $l_2 = \neg r_2$. After which, the remaining clauses become
$(\neg r_1 \vee \neg r_2 \vee \delta \vee x)$ and $(r_1 \vee r_2 \vee \alpha \vee x)$.
Then, we must have $r_1 = \neg r_2$, and we can remove one of the linked variable, an additional variable
from $C_2$ and $x$. This gives us a change of measure of $6\times0.8823-2\times0.1177$. Note that we require that
one of the two clauses to be at least length 4. When both clauses are length 3, by default, we have
already treat all variables to have weight $0.8823$, hence we ignore such cases. This gives us a branching
factor of at most $\tau(5\times0.8823,6\times0.8823-2\times0.1177)=1.1580$. We eliminate cases like this and
list other permutations in Case 2 here. 

Let $s$ be a variable not appearing in the clauses that we are discussing about. We redefine Case 2.1 and Case 2.2,
while keeping the Normal Case as before. Recall that Case 2 deals with the fact that some of the neighbours of $x$
have weight $0.8823$ instead.

Case 2.1 : We have clauses $C_1 = (a_1 \vee a_2 \vee ... \vee x)$ ,
$C_2 = (b_1 \vee b_2 \vee ... \vee x)$, 
$C_3 = (c_1 \vee c_2 \vee ... \vee x)$,
$(s \vee a_1 \vee b_1)$ and  $(s \vee b_2 \vee c_1)$. Note that some clause $C_2$
has two variables as neighbours of $s$. For the proof below, we will use this
(a clause having two variables in it as neighbours of s) notation to denote the
worst case. Note that any of the variables in Case 2.1 can be on further 3-literal
clauses of any permutation. 

Case 2.2 : No such $s$ occurs where $s$ appears in two 3-literal clauses, and
the neighbours of $s$ are the neighbours of $x$. Here, we consider two neighbours
of $x$ as being neighbours of a new variable not appearing in $C_1$, $C_2$ and $C_3$, or
Standalone 3-literal clauses here whichever gives the worst case.

For example, we have $(x \vee v_1 \vee v_2 \vee v_3)$, $(x \vee v_4 \vee v_5 \vee v_6)$ and $(x \vee v_7 \vee v_8 \vee v_9)$.
Here, we consider 3-literal clauses such as $(s \vee v_1 \vee v_4)$. Here $s$ appears
only once for this case. Similar to Case 2.2 of the previous section.

In Case 2, $s$ cannot appear in the third 3-literal clause, else Line 10 of the algorithm would have already
handled it. Therefore, the new variable $s$ can appear in at most two 3-literal clauses. Our cases here are complete.

\begin{lemma}
The time complexity of branching heavy variables is $O(1.1668^n)$.
\end{lemma}
\begin{proof}
Let $x$ be a heavy variable. Given $(l_1,l_2,l_3)$, then there are $|l_1| + |l_2| + |l_3| - 2$ unique
variables in these clauses. From the previous line of the algorithm, we know that any clause $C_1$ and $C_2$
must be that $|C_1 \cap C_2|\leq1$. Let $h$ denote the number of
3-literal clauses as shown in Case 2.2 above. We will give the Normal case, Case 1 (only for $(4,4,4)$), Case
2.1 and Case 2.2.  For Case 2.1, we will treat all variables as having weight $0.8823$ due to all the 
possibilities and permutations of 3-literal clauses that can occur here.
In addition, we handle the cases in the following order: $(3,3,\geq4)$, then $(3,\geq4,\geq4)$ etc.

\begin{itemize}
\item  $(3,3,\geq4)$. 
We'll first start with $(3,3,4)$. Branching $x=1$ will allow us to remove all the variables in this case,
with a change in measure of $5\times0.8823$ + $3$. When $x=0$, we will have a change in
measure of $3\times0.8823-2\times0.1177$, and when the 4-literal clause drops to a 3-literal clause,
another $3\times0.1177$. This gives $\tau(5\times0.8823+3,3\times0.8823+0.1177)=1.1591$. 

If Case 2.1 occurs, then we branch $x=1$ and $x=0$. When $x=1$, all the literals in the 3 clauses are
assigned $0$. This means that we can also remove $s$. We have a change of measure of 
$9\times0.8823$ here. On the other hand, when $x=0$, we link up the variables in the 3-literal clauses
that $x$ is in, giving us $3\times0.8823-2\times0.1177$. This gives us a branching factor of 
$\tau(9\times0.8823,3\times0.8823-2\times0.1177)=1.1620$. 

If Case 2.2 occurs, then we can have at most three 3-literal clauses, with each literal (apart from $x$) in the 4-literal clause 
appearing in a seperate 3-literal clause (appearing as Standalone 3-literal clause as being the worst).
We branch $x=1$ and $x=0$. Then the branching factor is given as 
$\tau(5\times0.8823+h\times(0.8823+0.7646)+(3-h),3\times0.8823+(3-h)\times0.1177-2\times0.1177)$, $h\in\{1,2,3\}$,
which is at max of $1.1540$ when $h=1$. This completes the case for $(3,3,4)$.

 Next, we deal with $(3,3,\geq5)$. For such a case, when $x=1$, we remove all variables, which gives us a change
of measure of $5\times0.8823+4$. When $x=0$, we remove $x$ and link up the two variables in the 3-literal clauses.
This gives us a change of $3\times0.8823-2\times0.1177$. The branching factor for this case would be 
$\tau(5\times0.8823+4,3\times0.8823-2\times0.1177)=1.1562$. If Case 2.1 or 2.2 applies here, 
then we give an upper bound to these cases by treating all variables as having weight $0.8823$.
When $x=1$, we remove all 9 variables, this gives us $9\times0.8823$. 
On the other hand, when $x=0$, we have $3\times0.8823-2\times0.1177$. This gives us at most
$\tau(9\times0.8823,3\times0.8823-2\times0.1177)=1.1620$. This completes the case for $(3,3,\geq5)$ and hence
$(3,3,\geq4)$. 

\item $(3,\geq4,\geq4)$. 
We start with $(3,4,4)$, then $(3,4,\geq5)$ and then $(3,\geq5,\geq5)$. 
Branching $x=1$ will allow us to remove all the variables, this
gives us a change of measure of $6+3\times0.8823$. On the other hand, branching $x=0$,
we can factor in a change of measure of $2\times0.8823-0.1177 + 6\times0.1177$. This gives us a branching
factor of $\tau(6+3\times0.8823,2\times0.8823+5\times0.1177)=1.1551$. 

For Case 2.1, the worst case happens when 
we have two variables in any of the 4-literal clauses as neighbours of $s$. Branching $s=1$ will allow us to remove
7 variables, where one of which is via linking of a variable in a 3-literal clause, giving us $7\times0.8823-0.1177$. When $s=0$, we remove
$x$, $s$ and 2 variables via linking in the 3-literal clause, giving us $4\times0.8823-2\times0.1177$. This gives  
$\tau(7\times0.8823-0.1177,4\times0.8823-2\times0.1177)=1.1653$. 

For Case 2.2, we can have at most three 3-literal
clauses appearing across the two 4-literal clauses. Then branching $x=1$ and $x=0$ gives us 
$\tau(3\times0.8823+h\times(3\times0.8823)+2\times(3-h),2\times0.8823-0.1177+(3-h)\times2\times0.1177)$, $h \in \{1,2,3\}$, 
which is at max of $1.1571$ when $h=3$. This completes the case for $(3,4,4)$.

For $(3,4,\geq5)$, branching $x=1$ will allow us to remove all variables, representing a change in measure of $3\times0.8823+7$. 
On the other hand, branching $x=0$ will allow us to remove $x$, link a variable in the 3-literal clause
and factor in the change in measure for the 4-literal clauses. This gives us 
$\tau(3\times0.8823+7,2\times0.8823+2\times0.1177)=1.1547$. 

For Case 2.1 and Case 2.2, we can find a variable $s$
that does not appear in any of the clauses. We give an upper bound for this case by treating all variables as having weight $0.8823$.
When $x=1$, we remove all 10 variables and $s$. This
gives us $11\times0.8823$. When $x=0$, we remove $x$ and link up the other variable in the 3-literal clause, giving us
$2\times0.8823-0.1177$. This gives us a branching factor of at most $\tau(11\times0.8823,2\times0.8823-0.1177)=1.1666$.
This completes the case for $(3,4,\geq5)$. 

Finally, for the case of $(3,\geq5,\geq5)$, we give an upper bound for this case by 
treating all the variables as having weight $0.8823$, to deal with the Normal Case, Case 2.1
and 2.2 at the same time. Branching $x=1$ gives us a change of measure of $11\times0.8823$.
When $x=0$, this gives us $2\times0.8823-0.1177$. Putting them together, we have a branching factor of at
most $\tau(11\times0.8823,2\times0.8823-0.1177)=1.1666$ for this case.
This completes the case for $(3,\geq5,\geq5)$ and hence $(3,\geq4,\geq4)$. 

\item $(4,4,4)$.
When $x=1$, we remove all variables. This gives us a change of measure of $10$.
On the other hand, when $x=0$, we have a change of measure of $1+9\times0.1177$.
This gives us a branching factor of $\tau(10,1+9\times0.1177)=1.1492$. If Case 1 occurs,
then we have at most $\tau(10,1+0.8823)=1.1548$. 

When Case 2.1 occurs, then one of the 4-literal clause must have 2 variables in it that are neighbours to $s$.
Suppose we have $(x \vee a_1 \vee a_2 \vee a_3)$, $(x \vee b_1 \vee b_2 \vee b_3)$, 
$(x \vee c_1 \vee c_2 \vee c_3)$, $(s \vee a_1 \vee b_1)$ and $(s \vee b_2 \vee c_1)$. 
Then we branch $b_1=1$ and $b_1=0$. When $b_1=1$, then $s=a_1=b_2=b_3=x=0$.
Since $s=b_2=0$, then $c_1=1$. Therefore, we must have $c_2=c_3=0$ and we can link up
$a_2 = \neg a_3$. Now, $x$ must have weight 1, else earlier cases would have handled it.
This gives a change of measure of $9\times0.8823+1$. On the other hand, when $b_1=0$,
we link up $a_1 = \neg s$ (no increase in measure here since $s$ is still in another 3-literal clause), 
$x$ will drop in weight, giving us a change of measure of $2\times0.8823+0.1177$.
This gives a branching factor of at most $\tau(9\times0.8823+1,2\times0.8823+0.1177)=1.1668$. 

When Case 2.2 arises, then the worst case happens when we have three 3-literal clauses
appearing between two of the 4-literal clauses. We branch $x=1$ and $x=0$.
When $x=1$, we remove all the variables in the clauses, along with 3 new variables not in these
3 clauses containing $x$. This gives us $9\times0.8823+4$. On the other hand, when 
$x=0$, one of the 4-literal clause drops to a 3-literal clause, giving us 
$1+3\times0.1177$. This gives us a branching factor of 
$\tau(9\times0.8823+4,1+3\times0.1177)=1.1577$.
This completes the case for $(4,4,4)$. 

\item $(4,4,\geq5)$. 
When $x=1$, we remove all 11 variables. When $x=0$, we remove $x$ and factor
in the change of measure from the 4-literal clauses, giving us $1+6\times0.1177$.
This gives us a branching factor of $\tau(11,1+6\times0.1177)=1.1509$. 

If Case 2.1 occurs, then two variables from a 4-literal or 5-literal clause is a neighbour to $s$, 
Suppose the 2 variables appear in the 4-literal clause. We can then choose a literal $a$ in this 4-literal clause 
to branch such that we have a literal in the 5-literal clause being assigned 1 (same technique as above) when $a=1$.
This allows us to remove 10 variables of weight $0.8823$ and one of weight 1 ($x$), giving us 
$10\times0.8823+1$. On the other hand, when $a=0$, we link up a variable with $s$ in the 
3-literal clause, and then factor in the drop of weights for $x$, giving us 
$2\times0.8823+0.1177$. This gives us $\tau(10\times0.8823+1,2\times0.8823+0.1177)=1.1566$. 
If the two variables appear in the 5-literal clause, then again, we apply the same technique. 
This gives us $\tau(10\times0.8823+1,2\times0.8823)=1.1608$.

If Case 2.2 applies, then there 
are at most three 3-literal clauses between the two 4-literal clauses. 
Then our branching factor is given as 
$\tau(h\times(3\times0.8823)+5+2\times(3-h),1+2\times(3-h)\times0.1177)$, $h\in \{1,2,3\}$, 
which is at max of $1.1637$ when $h=3$. This completes the case for $(4,4,\geq5)$. 

\item $(4,\geq5,\geq5)$. 
When $x=1$, we remove all 12 variables. When $x=0$, we remove $x$ and factor in the 
change of measure of $1+3\times0.1177$. Therefore, we have
$\tau(12,1+3\times0.1177)=1.1551$. 

When Case 2.1 occurs, then follow the same
technique as given in the previous case to get an upper 
bound of $\tau(10\times0.8823+1,2\times0.8823)=1.1608$.

\item $(\geq5,\geq5,\geq5)$. 
When $x=1$, we remove 13 variables and when $x=0$, we remove only $x$. 
This gives us $\tau(13,1)=1.1632$. 

If Case 2.1 occurs, 
follow the same technique as given in $(4,4,\geq5)$ to get an upper bound of
$\tau(10\times0.8823+1,2\times0.8823)=1.1608$.

For Case 2.2,
then worst case occurs when every variable in $(\geq5,\geq5,\geq5)$ has weight $1$, which gives $1.1632$ (Normal Case). 
This is because when $x=0$, we can only remove $x$ and not factor in any other change in measure. On the other hand,
when any of the variables have weight $0.8823$, this means we can remove additional variables when $x=1$.
Therefore, giving us a lower branching factor.
This completes the case for $(\geq5,\geq5,\geq5)$.
\end{itemize}

Hence, Line 14 of the algorithm runs in $O(1.1668^n)$ time.
\end{proof}

Therefore, putting all the lemmas together, we have the following result :
\begin{theorem}
The algorithm runs in $O(1.1674^n)$ time.
\end{theorem}

In summary, we proposed a DPLL style algorithm to solve the XSAT problem in $O(1.1674^n)$. Prior to this work,
the current state of the algorithm is another DPLL style algorithm which runs in $O(1.1730^n)$. The novelty of our algorithm lies
on the design of a nonstandard measure to help us to tighten our analysis further. However, this has led to some additional
cases that we have to analyse. Perhaps a question for interested readers would be : Is it possible to design a simple nonmeasure to 
either improve the worst case bound further ? Or to cut down the number of cases that we need to analyse.


\begin{thebibliography}{99}

\bibitem{Sch78} T.J. Schaefer.
\newblock {\em The complexity of satisfiability problems.}
\newblock Proc. STOC 1978. ACM (1978), p. 216-226.

\bibitem{Cook71} S. Cook.
\newblock {\em The Complexity of Theorem Proving Procedures.}
\newblock Proc. 3rd Annual ACM Symposium on Theory of Computing (STOC). pp. 151–158.

\bibitem{Silva08} J. Marques-Silva.
\newblock {\em Practical applications of Boolean Satisfiability.}
\newblock 2008 9th International Workshop on Discrete Event Systems. IEEE, 2008, pp. 74–80.

\bibitem{SS81} R. Schroeppel,A. Shamir.
\newblock {\em A $T=O(2^{n/2})$, $S=O(2^{n/4})$ algorithm for certain NP-complete problems.}
\newblock SIAM J. Comput. 10(3) (1981) 456-464.

\bibitem{MSV81} B. Monien, E.Speckenmeyer, O. Vornberger.
\newblock {\em Upper bounds for covering problems.}
\newblock  Methods Oper. Res. 43 (1981) 419-431.

\bibitem{BMS05} J.M. Byskov, B.A. Madsen, B. Skjernaa.
\newblock {\em New algorithms for exact satisfiability} 
\newblock Theoretical Computer Science 332, no. 1-3 (2005): 515-541.

\bibitem{D06} V. Dahll\"of .
\newblock {\em Exact Algorithms for Exact Satisfiability Problems.}
\newblock Link\"oping Studies in Science and Technology, PhD Dissertation no 1013, 2006.

\bibitem{FGK09} F.V. Fomin, F. Grandoni and D. Kratsch 
\newblock {\em A measure and conquer approach for the analysis of exact algorithms.}
\newblock Journal of the ACM (JACM) 56, no. 5 (2009): 25.

\bibitem{FK10} F.V.\ Fomin and D. Kratsch.
\newblock {\em Exact Exponential Algorithms.}
\newblock Texts in Theoretical Computer Science. An
EATCS Series. Springer, Berlin, Heidelberg, 2010.

\bibitem{Kul99} O. Kullmann.
\newblock {\em New methods for 3-SAT decision and worst-case analysis.}
\newblock Theoretical Computer Science, 223(1-2):1-72, 1999. 

\bibitem{DPLL60} Davis, Martin; Logemann, George; Loveland, Donald
\newblock {\em  A Machine Program for Theorem Proving}
\newblock  Communications of the ACM. 5 (7): 394–397.

\end{thebibliography}
\end{document}